\documentclass[copyright,creativecommons]{eptcs}

\providecommand{\event}{``Post-Proceedings of the 5th International Workshop on Logical Frameworks and Meta-languages: Theory and Practice, 2010''}

\usepackage{breakurl}             

\usepackage{latexsym}
\usepackage{amsmath}
\usepackage{amssymb}
\usepackage{amsfonts}
\usepackage{proof}
\usepackage{code}
\usepackage{cdsty}
\usepackage{url}
\usepackage{ifthen}
\usepackage{stmaryrd}
\usepackage{beluga}

\ifdefined\LONGVERSION
  \relax
\else
\newcommand{\LONGVERSION}[1]{}
\newcommand{\SHORTVERSION}[1]{#1}
\fi

\newtheorem{corollary}[theorem]{Corollary}

\def\squareforqed{\ensuremath{\Box}}
\def\qed{\ifmmode\squareforqed\else{\unskip\nobreak\hfil
\penalty50\hskip1em\null\nobreak\hfil\squareforqed
\parfillskip=0pt\finalhyphendemerits=0\endgraf}\fi}

\newenvironment{proof}[1][]{\noindent\ifthenelse{\equal{#1}{}}{{\it
      Proof.}}{{\it Proof #1.}}\hspace{2ex}}{\qed\bigskip}
\newenvironment{proof*}[1][]{\noindent\ifthenelse{\equal{#1}{}}{{\it
      Proof.}}{{\it Proof #1.}}\hspace{2ex}}{\bigskip}

\usepackage{srcltx}
\usepackage{listings}
\newdimen\zzlistingsize
\newdimen\zzlistingsizedefault
\zzlistingsize=\zzlistingsizedefault
\global\def\InsideComment{0}
\newcommand{\Lstbasicstyle}{\fontsize{\zzlistingsize}{1.05\zzlistingsize}\ttfamily}
\newcommand{\keywordFmt}{\fontsize{0.9\zzlistingsize}{1.0\zzlistingsize}\bf}
\newcommand{\smartkeywordFmt}{\if0\InsideComment\keywordFmt\fi}
\newcommand{\commentFmt}{\def\InsideComment{1}\fontsize{0.95\zzlistingsize}{1.0\zzlistingsize}\rmfamily\slshape}

\newlength{\zzlstwidth}
\newcommand{\setlistingsize}[1]{\zzlistingsize=#1%
\settowidth{\zzlstwidth}{{\Lstbasicstyle~}}}
\setlistingsize{\zzlistingsizedefault}

\newcommand{\bla}{\ensuremath{\mbox{$$}}} 

\newcommand{\der}{\,\vdash}

\newcommand{\length}[1]{|#1|}

\def\lv{\mathopen{{[\kern-0.14em[}}}    
\def\rv{\mathclose{{]\kern-0.14em]}}}   

\newcommand{\dens}[1]{\mathopen{[\kern-0.3ex[}#1\mathclose{]\kern-0.3ex]}}
\newcommand{\denk}[2]{\mathopen{\{\kern-0.3ex|}#1\mathclose{|\kern-0.3ex\}}_{#2}}
\def\lo{\mathopen{{\lceil\kern-0.25em\lceil}}}    
\def\ro{\mathclose{{\rfloor\kern-0.25em\rfloor}}}

\def\ltox#1{\buildrel\raise1pt\hbox{$\scriptstyle#1$}\over\longrightarrow}

\def\tocolow{\buildrel\raise-5pt\hbox{$\scriptscriptstyle+$}\over\rightarrow}

\newcommand{\abbrev}[1]{#1} 

\newcommand{\ie}{\abbrev{i.\,e.}}

\newcommand{\para}[1]{\paragraph*{\it#1}}
\newcommand{\paradot}[1]{\para{#1.}}

\newcommand{\mgoal}[1][]{\mbox{goal\ifthenelse{\equal{#1}{}}{}{~#1}}}

\newcommand{\ru}[2]{\dfrac{\begin{array}[b]{@{}c@{}} #1 \end{array}}{#2}}

\newcommand{\cxt}{\ctx}  
\newcommand{\mcxt}{\mctx} 
\newcommand{\shiftby}[1]{\shift^{#1}}
\newcommand{\Shiftby}[1]{\Shift^{#1}}
\newcommand{\mshiftby}{\Shiftby}
\newcommand{\mId}{\Shiftby 0} 
\newcommand{\esub}[1]{[#1]}
\newcommand{\esubp}[2]{[#1](#2)}
\newcommand{\msubp}[2]{\msub{#1}(#2)}
\newcommand{\sgsub}[1]{\esub{\shiftby 0,#1}}
\newcommand{\dsub}[3]{\esub{#1}{\msub{#2}{#3}}} 
\newcommand{\dsubp}[3]{\esub{#1}{\msubp{#2}{#3}}} 
\newcommand{\fun}[1]{\Pi\,#1.\,}

\newcommand{\abox}[1]{\multicolumn 1 {@{}p{\textwidth}@{}} {#1}}
\newcommand{\aboxiii}[1]{\multicolumn 3 {@{}p{\textwidth}@{}} {#1}}
\newcommand{\aleq}{\buildrel\mathsf{l}\over\sim}

\newcommand{\alr}{\buildrel\mathsf{r}\over\sim}
\newcommand{\alreq}[4]{\shiftEnv{#1}{#3} \alr \shiftEnv{#2}{#4}}
\newcommand{\alw}{\buildrel\mathsf{w}\over\sim}
\newcommand{\aln}{\buildrel\mathsf{n}\over\sim}

\newcommand{\als}{\buildrel\mathsf{s}\over\sim}

\newcommand{\lift}[1]{\esub{\shiftby 1}{#1},x_1}

\newcommand{\tlookup}{\mathsf{lookup}}
\newcommand{\wlookup}[2]{\tlookup~#1~{x_{#2}}}
\newcommand{\wlookupp}[2]{\wlookup{(#1)}{#2}}

\newcommand{\tLookup}{\mathsf{Lookup}}
\newcommand{\wmlookup}[2]{\tLookup~#1~{X_{#2}}}
\newcommand{\wmlookupp}[2]{\wmlookup{(#1)}{#2}}
\newcommand{\twhnf}{\mathsf{whnf}}
\newcommand{\twhnfp}[1]{\twhnf~(#1)}
\newcommand{\whnf}[3]{\twhnf~\dsub{#1}{#2}{#3}}
\newcommand{\whnfp}[3]{\whnf{#1}{#2}{(#3)}}
\newcommand{\twsub}{\mathsf{env}}
\newcommand{\wsub}[3]{\twsub~#1~#2~#3}
\newcommand{\wsubp}[3]{\wsub{#1}{#2}{(#3)}}

\newcommand{\twmsub}{\mathsf{Env}}
\newcommand{\wmsub}[2]{\twmsub~#1~#2}

\newcommand{\shiftEnv}[1]{\esub{\shiftby{#1}}}
\newcommand{\shiftEnvp}[2]{\esub{\shiftby{#1}}{(#2)}}
\newcommand{\tshift}{\mathsf{shift}}
\newcommand{\shiftClos}[1]{\tshift^{#1}\,}
\newcommand{\shiftClosp}[2]{\shiftClos{#1}{(#2)}}
\newcommand{\shiftNe}{\shiftClos}
\newcommand{\shiftNep}{\shiftClosp}
\newcommand{\wapp}[2]{#1 \mathrel{@} #2}

\newcommand{\sid}{\shiftby 0}

\newcommand{\jinf}{\rightrightarrows}
\newcommand{\jchk}{\leftleftarrows}

\newcommand{\cempty}{\mathord{\cdot}}

\newcommand{\EL}{\mathcal{E}\kern-0.2ex\ell}

\newcommand{\x}{\mathsf{x}}
\newcommand{\xdel}[1][]{\ifthenelse{\equal{#1}{}}{\x_\Delta}{\x_{\Delta+#1}}}

\newcommand{\VDash}{\mathrel{\mathord{|}\kern-0.15ex\mathord{\models}}}

\title{Explicit Substitutions for Contextual Type Theory}
\author{Andreas Abel
\institute{Theoretical Computer Science, Ludwig-Maximilians-University
  Munich, Germany}
\email{andreas.abel@ifi.lmu.de}
\and
  Brigitte Pientka
\institute{School of Computer Science, McGill University, Montreal, Canada}
\email{bpientka@cs.mcgill.ca}
}
\def\titlerunning{Explicit substitutions for contextual type theory}
\def\authorrunning{A. Abel \& B. Pientka}

\begin{document}
\maketitle

\begin{abstract}
In this paper, we present an explicit substitution calculus
which distinguishes between ordinary bound variables and
meta-variables. Its typing discipline is derived from contextual modal
type theory. We first present a dependently typed lambda calculus with
explicit substitutions for  ordinary variables and explicit
meta-substitutions for meta-variables.  We then present a weak head
normalization procedure which performs both substitutions lazily and
in a single pass thereby combining substitution walks for the two
different classes of variables. Finally, we describe a bidirectional
type checking algorithm which uses weak head normalization and prove
soundness. 
\\
Keywords: 
Explicit substitutions,
Meta-variables,
Logical framework,
Contextual modal type theory
\end{abstract}

\section{Introduction}
Over the last decade, reasoning and programming with dependent types 
has received wide attention and several systems provide implementations
for dependently typed languages (see for example Agda \cite{boveDybjerNorell:tphols09,norell:PhD}, Beluga \cite{Pientka:PPDP08,pientkaDunfield:ijcar10}, Delphin \cite{Schuermann:ESOP08,poswolskySchuermann:delphin}, Twelf \cite{Pfenning99cade}, etc). 

As dependent types become more accepted, it is interesting to better
understand how to implement such systems efficiently.
While all the systems mentioned support type checking and moreover provide implementations 
supporting type reconstruction for dependent types, there is a
surprising lack in documentation and gap in modelling the theoretical foundations of these implementations. This makes it hard to reproduce some
of the ideas, and prevents them from being widely accessible to a
broader audience.

A core question in the implementations for dependently typed systems
is how to handle substitutions. Let us illustrate the problem in the
setting of contextual modal type theory \cite{Nanevski:ICML05}, where
we not only have ordinary $\Pi$-types to abstract over ordinary
variables $x$ but also $\Pibox$-types which allow us to abstract over
meta-variables $X$, and we find the following two elimination rules:

\[
\begin{array}{l}
\infer{\Delta ; \Gamma \vdash M\;N : [N/x]B}{
\Delta ; \Gamma \vdash M : \Pi x{:}A.B  &
\Delta ; \Gamma \vdash N : A}
\quad \quad  
\infer{\Delta ; \Gamma \vdash M\;(\Psihat.N) : \msub{\Psihat.N/X}B}{
\Delta ; \Gamma \vdash M : \Pibox X{:}A[\Psi].B  &
\Delta ; \Psi \vdash N : A}
\end{array}
\]

In the $\Pi$-elimination rule, we do not want to apply
the substitution $N$ for $x$ in the type $B$ eagerly during type
checking, but accumulate all the individual substitutions and apply
them simultaneously, if necessary. Similarly, in the
$\Pibox$-elimination rule, we do not want to replace eagerly the
meta-variable $X$ with $N$ in the type $B$ but accumulate all
meta-substitutions and also apply them simultaneously. 
In fact, we would like to combine substitution walks for
meta-variables and ordinary variables, and simultaneously apply
ordinary substitution and meta-substitutions to avoid multiple
traversals. This will allow us potentially to detect that two terms
are not equal without actually performing a substitution, and in the
case of a de Bruijn numbering scheme for variables, we would like to
avoid unnecessary renumbering.  

Explicit substitutions go back to Abadi et al
\cite{abadiCardelliCurienLevy:jfp91} and
are often central when implementing core algorithms such as type
checking or  higher-order unification 
\cite{dowekHardinKirchner:infcomp00}. 
Many existing implementations of proof assistants such as the Twelf
system, Delphin, Beluga, Agda or $\lambda$Prolog use explicit
substitutions to combine substitution walks for ordinary
variables. A different approach with the same goal of handling
substitutions efficiently is the suspension calculus 
\cite{Nadathur:TCS98,liangNadathurQi:jar05}.

 However, meta-variables are often modeled via references
thereby avoiding the need to explicitly propagate substitutions for meta-variables. Yet
there are multiple reasons why we would like to treat meta-variables
non-destructively and be able to handle meta-substitutions
explicitly. First, such implementations may be easier to
maintain and may be more efficient. Second, in several applications
we need to abstract over the remaining free meta-variables in the most
general solution found by higher-order unification.  For example in type
reconstruction we need to store a closed most general type or  in tabled
higher-order logic programming \cite{Pientka03phd} we want to
store explicitly the answer substitution for the meta-variables occurring in a
query. Abstraction can be expensive since we need to first
traverse a term including the types of all the meta-variables
occurring in it and  collect all references to
meta-variables. Subsequently, we again need to traverse the term
including the types of meta-variables and compute their appropriate de
Bruijn index.  A non-destructive implementation of unification could
avoid this explicit abstraction step. To achieve a practical,
non-destructive implementation of unification, understanding the interaction of
ordinary substitutions with meta-substitutions and handling both
lazily is crucial.

While meta-variables are often only introduced internally, i.e., there
is no abstraction over meta-variables using a $\Pibox$-type, some
languages such as Beluga have taken the step to distinguish ordinary
bound variables and meta-variables already in the source
language. Consequently, we find different classes of bound variables,
bound ordinary variables and bound meta-variables, and different
types, $\Pi$- and $\Pibox$-types. When 
type-checking Beluga programs, we would like to combine substitution
walks for these different classes. Understanding how these two
substitutions interact is also crucial for type reconstruction in this
setting, since omitted arguments may depend on both kinds of variables.

In this paper, we revisit the ideas of explicit substitutions where we
combine substitutions for ordinary variables and meta-variables. In
particular, we describe an explicit substitution calculus with
first-class meta-variables inspired by contextual modal type theory
\cite{Nanevski:ICML05}. We first present a dependently typed lambda
calculus with explicit substitutions for  ordinary variables and
explicit meta-substitutions for meta-variables. We omit here the
ability to abstract explicitly over meta-variables which is a
straightforward addition and concentrate on the interaction of
ordinary substitutions and meta-substitutions. We then present a weak
head normalization procedure which performs both substitutions lazily
and in a single pass thereby combining substitution walks for 
the two different classes of variables. Finally, we give an algorithm
for definitional equality and present a bidirectional type checking
algorithm which employs weak head normalization and show soundness. 
In the future, we plan to use the presented calculus as a
foundation for implementing the Beluga language which supports
programming and reasoning with formal systems specified in the logical
framework LF.

\section{The Calculus: Syntax, Typing, and Equality}

Let us first introduce the grammar and typing rules for the
dependently typed $\lambda$-calculus  with meta-variables based on
the ideas in \cite{Nanevski:ICML05}. The system we consider
is an extension of the logical framework LF with first-class
meta-variables. We design the calculus 
as an extension of previous explicit substitution calculi such as
\cite{abadiCardelliCurienLevy:jfp91,dowekHardinKirchner:infcomp00}. These calculi only support ordinary
substitutions but not at the same time meta-substitutions. 

Our calculus supports general closures on the type and term
level. Meta-variables (which sometimes are also called contextual
variables) are written as $X$. 
Typically, meta-variables occur as a 
closure $\clo{X}{\sigma}$, but we will treat this as a special case of
the general closure $\clo{N}{\sigma}$. 

To provide a compact representation of the typing rules, we follow the
tradition of pure type systems and introduce sorts and expressions
where sorts can be either $\lfkind$ or $\lftype$ and expressions
include terms, types and kinds. A single syntactic category of
expressions helps us avoid duplication in the typing and equality
rules for closures $\esub \sigma E$ and 
$\msub \theta E$. We will write $M$, $A$, $K$, if indeed expressions
can only occur as terms $M$, types $A$ or kinds $K$.

\[ 
\begin{array}{@{}llcl@{}r@{}}
\mbox{Sorts} & s & \bnfas & \lfkind \mid \lftype \\
\mbox{Expressions} & E,F & \bnfas & \multicolumn 2 {l@{}} {
s \mid a \mid \fun E F
 \mid x_n \bnfalt X_n \bnfalt \lam{x}{E}
\bnfalt \app{F}{E} \bnfalt \clo{E}{\sigma} \bnfalt \cclo{E}{\theta} 
} \\[0.5em]
\hline\\[-0.75em]
\multicolumn{4}{@{}l@{}}{\mbox{Special cases of expressions:}} \\[0.5em]
\mbox{Kinds} & K & \bnfas & \lftype \bnfalt \Pi A.K \bnfalt
\clo{K}{\sigma} \bnfalt \cclo{K}{\theta} \\
\mbox{Types} & A,B & \bnfas & a \mid A\;M \bnfalt \Pi A.B
\bnfalt \clo{A}{\sigma}  \bnfalt \cclo{A}{\theta} \\ 
\mbox{Terms} & M,N & \bnfas &  x_n \bnfalt X_n \bnfalt \lam{x}{M}
\bnfalt \app{M}{N} \bnfalt \clo{N}{\sigma} \bnfalt \cclo{M}{\theta} 
& (n \geq 1)
\\[0.5em]
\hline \\[-0.75em]
\mbox{Substitutions} & \sigma, \tau & \bnfas & \shift^n \bnfalt 
  \sigma, M  \bnfalt \clo{\sigma}{\tau}  \bnfalt \cclo{\sigma}{\theta}
& (n \geq 0)
\\[0.5em]
\mbox{Meta-substitutions} 
& \theta & \bnfas & \Shift^n \bnfalt
\theta,M \bnfalt \cclo{\theta'}{\theta}
& (n \geq 0)
\\[0.5em]
\mbox{Contexts} & \Gamma,\Psi & \bnfas & \edot \bnfalt \Psi, A  
\\
\mbox{Meta-contexts} & \Delta & \bnfas &  \edot \bnfalt \Delta,
\cdec{\Psi}{A} 

\end{array} 
\]
Constants are denoted by letter $a$, their types/kinds are recorded in
a global well-formed signature $\Sigma$.
We have two different de Bruijn indices $x_n$ and $X_n$ ($n \geq 1$), 
one for numbering bound variables and one for numbering meta-variables.  
$x_n$ represents the
de Bruijn number $n$ and stands for an ordinary bound variable, while
$X_n$ represents the de Bruijn number $n$ but stands for a
meta-variable. Due to the two kinds of substitutions, we also have two
kinds of closures; the closure of an expression with an ordinary substitution
$\sigma$ and the closure of an expression with a meta-substitution
$\theta$. Following the treatment of meta-variables in
\cite{Nanevski:ICML05}, we describe the type of a meta-variable as
$\cdec{\Psi}{A}$ which stands for a meta-variable of type $A$ which
may refer to variables in $\Psi$. 

Meta-substitutions provide a term $M$ for a meta-variable $X$ of type
$\cdec{\Psi}{A}$. Note that $M$ does not denote a closed term, but a
term of type $A$ in the context $\Psi$ and hence may refer to
variables from $\Psi$. In previous presentations where we use names
for variables, we hence wrote $\Psihat.M/X$ to be able to rename the
variables in $M$ appropriately. Because bound variables
are represented using de Bruijn indices in this paper, we simply write
$M/X$ but keep in mind that $M$ is not necessarily closed.

Our calculus also features closures on the level of substitutions and
meta-substitutions. For example, we allow the closure $[\sigma]\tau$
which will allow us to lazily treat ordinary substitution
composition and the closure $\msub{\theta}\sigma$ which will postpone
applying $\theta$ to the ordinary substitution $\sigma$. Similarly,
the closure $\msub{\theta}\theta'$ for meta-substitutions allows us to
lazily compose meta-substitutions. We note the absence of a closure
$[\sigma]\theta$. Applying an ordinary substitution $\sigma$ to a
meta-substitution $\theta$ simply reduces to $\theta$, since all
objects in the meta-substitution are closed objects and cannot be 
affected by $\sigma$. It is hence not meaningful to include a closure
$[\sigma]\theta$. 
We also do not introduce a closure of a context $\Psi$ and a
meta-substitution $\theta$. Instead we define $\msub \theta \Psi$ eagerly
by simply pushing the meta-substitution $\theta$ to each
declaration as follows: $\msub \theta \cempty = \cempty$ and $\msubp \theta {\Psi,A} = \msub
\theta \Psi,\, \msub \theta A$.  The length of a context $\Gamma$ is
denoted by $|\Gamma|$ and likewise $|\Delta|$ for meta-contexts.

\begin{figure}[htbp]
\centering
\[
\begin{array}{c}
\multicolumn{1}{l}{\mbox{Expressions }}\\
\ru{\Delta \der \Gamma \cxt
  }{\Delta; \Gamma \der \lftype : \lfkind}
\qquad
\ru{\Delta \der \Gamma \cxt \qquad \Sigma(a) = K
  }{\Delta; \Gamma \vdash a : K}
\qquad
\ru{
    \Delta; \Gamma, A \vdash E : s
  }{\Delta;\Gamma \vdash \Pi A.E : s}
\\[1.25em]
\infer{\Delta; \Gamma, A \vdash x_1 : \esub {\shiftby 1} A}
      {\Delta; \Gamma \der A : \lftype}
\qquad
\infer{\Delta; \Gamma, B \vdash x_{n+1} : [\shift^1]A}
      {\Delta; \Gamma \vdash x_n : A & \Delta; \Gamma \der B : \lftype} 
\\[0.75em]
\infer{\Delta, \cdec{\Gamma}{A} ; \msub{\Shiftby 1}\Gamma 
         \vdash X_1 : \msub {\Shiftby 1} A}
      {\Delta; \Gamma \der A : \lftype}
\qquad
\infer{\Delta, \cdec{\Gamma'}{A'} ; \msub{\Shift^1}\Gamma \vdash
  X_{n+1} : \msub{\Shift^1}A}
      {\Delta; \Gamma \vdash X_n : A & \Delta; \Gamma' \der A' : \lftype}
\\[0.75em]
\infer{\Delta; \Gamma \vdash \lam{x}{M} : \Pi A. B}
      {\Delta; \Gamma, A \vdash  M : B \qquad
       \Delta; \Gamma, A \vdash B : \lftype}
\qquad
\infer{\Delta; \Gamma \vdash \app{E}{N} : [\shift^0, N]F}
       {\Delta; \Gamma \vdash E : \Pi A.F
        &
        \Delta; \Gamma \vdash  N : A}
\\[0.75em]
\infer{\Delta; \Gamma \vdash \clo{E}{\sigma} : \lfkind}
      {\Delta; \Gamma \vdash \sigma : \Psi & \Delta; \Psi \vdash E : \lfkind 
      }
\qquad 
\infer{\Delta; \Gamma \vdash \clo{E}{\sigma} : [\sigma]F}
      {\Delta; \Gamma \vdash \sigma : \Psi & \Delta; \Psi \vdash E : F 
      }
\\[0.75em]
\infer{\Delta; \cclo{\Gamma}{\theta} \vdash \cclo{E}{\theta} : \lfkind}
{\Delta \vdash \theta : \Delta' & 
 \Delta' ; \Gamma \vdash E : \lfkind}
\qquad 
\infer{\Delta; \cclo{\Gamma}{\theta} \vdash \cclo{E}{\theta} : \cclo{F}{\theta}}
{\Delta \vdash \theta : \Delta' & 
 \Delta' ; \Gamma \vdash E : F}
\qquad
\infer{\Delta; \Gamma \vdash E : F_2}{
\Delta; \Gamma \vdash E : F_1 & \Delta; \Gamma \vdash F_1 \equiv F_2 : s
}
\\[0.5em]
\multicolumn{1}{l}{\mbox{Contexts and meta-contexts}}\\
\infer{\vdash \cdot \mctx } {}
\qquad
\infer{\vdash \Delta, \cdec{\Psi}{A} \mctx}{
\Delta; \Psi \vdash A : \lftype
}  
\qquad\qquad
\infer{\Delta \vdash \cdot \ctx}{\der \Delta \mcxt}
\qquad
\infer{\Delta \vdash \Psi, A \ctx}{
\Delta; \Psi \vdash A : \lftype
}
\\[0.5em]
\multicolumn{1}{l}{\mbox{Ordinary substitutions}} \\[0.25em]
 \infer[]
       {\Delta; \Psi, \Gamma \vdash \shift^n : \Psi}
       {\Delta \der \Psi,\Gamma \cxt & |\Gamma| = n }
\qquad
 \infer[]
       {\Delta; \Gamma \vdash (\sigma, M) : (\Psi, A)}
       {\Delta; \Gamma \vdash \sigma : \Psi &
        \Delta; \Psi \der A : \lftype &
        \Delta; \Gamma \vdash M : [\sigma]A}
\\[0.75em]
\infer{\Delta; \Gamma \vdash \clo{\sigma}{\tau} : \Psi}
  {\Delta; \Gamma \vdash \tau : \Psi' & 
   \Delta; \Psi' \vdash \sigma : \Psi }
\qquad
\infer{\Delta; \cclo{\Gamma}{\theta} \vdash \cclo{\sigma}{\theta} :  \cclo{\Psi}{\theta}}
{\Delta \vdash \theta : \Delta' & 
  \Delta' ; \Gamma \vdash \sigma : \Psi  
}
\\
\multicolumn{1}{l}{\mbox{Meta-substitutions}} \\[0.25em]
\infer{\Delta, \Delta' \vdash \Shift^n : \Delta}{\der \Delta,\Delta' \mcxt
  & |\Delta'| = n }
\quad\quad
\infer{\Delta \vdash (\theta, M) : \Delta' , \cdec{\Gamma}{A}}
{\Delta \vdash \theta : \Delta' &
  \Delta'; \Gamma \der A : \lftype &
  \Delta; \msub{\theta}\Gamma \vdash M : \msub{\theta}A
}  
\\[0.75em]
\qquad
\infer{\Delta \vdash \cclo{\theta'}{\theta} : \Delta'}{
\Delta   \vdash \theta  : \Delta_0 &
\Delta_0 \vdash \theta' : \Delta' 
}
\end{array}
\]     
 \caption{Typing rules for explicit substitution calculus with first-class meta-variables}
 \label{fig:tprules}
\end{figure}


\subsection{Typing rules }

In contrast to \cite{harperPfenning:equivalenceLF},
we present the typing rules for LF in pure type system (PTS) style, to
avoid rule duplication (which would be substantial for the rules of
definitional equality given in the next section). We will use the
following judgments: 
\[
\begin{array}{lclp{10cm}}
                & \vdash & \Delta \mctx & Meta-context $\Delta$ is well-typed \\
\Delta          & \vdash & \Psi \ctx    & Context $\Psi$ is well-typed
\\
\Delta; \Gamma & \vdash & E : F      & Expression $E$ has ``type'' $F$\\
\Delta; \Gamma & \vdash & \sigma : \Psi & Substitution $\sigma$ has
domain $\Psi$ and range $\Gamma$\\
\Delta  & \vdash & \theta : \Delta' & Contextual Substitution $\theta$ has
domain $\Delta'$ and range $\Delta$
\end{array}
\]
The judgement $\Delta; \Gamma \der E : F$ subsumes the judgements
$\Delta; \Gamma \der M : A$ (term $M$ has type $A$), 
$\Delta; \Gamma \der A : K$ (type family $A$ has kind $K$) and
$\Delta; \Gamma \der K : \lfkind$ (kind $K$ is well-formed).

We present the typing rules in Figure~\ref{fig:tprules} as a type
assignment system for expressions $E$. 
To improve readability, we use the letters $M,N$ instead of $E$ when
we know that we are dealing with a term, and similarly $A,B$ for types and $K$
for kinds.

In the typing rule for $\lambda M$, the hypothesis $\Delta; \Gamma,A
\der B : \lftype$ prevents us to form a $\lambda$-abstraction on the type
level (for this, we would need $B : \lfkind$).  Lambda on the type
level does not increase the expressiveness 
\cite{adams:PhD,harperPfenning:equivalenceLF}.
Unlike the system in \cite{harperPfenning:equivalenceLF}, 
we do not assume that the
meta-context $\Delta$ and the context $\Gamma$ are well-formed, but
ensure that these are well-formed contexts by adding appropriate
typing premises to for example the typing rules for bound variables
and meta-variables. We establish separately that contexts are
well-formed (see Lemma \ref{lem:cwf} on page \pageref{lem:cwf}) and
that the inference rules are valid (see Theorem \ref{thm:synval} on
page \pageref{thm:synval}).

We concentrate here on explaining the typing rules for bound variables
and meta-variables. The typing rules for bound variables essentially
peel off one type declaration in the context $\Gamma$ until we
encounter the variable $x_1$. The typing premises guarantee that
the meta-context $\Delta$ and the context $\Gamma$ and the type $A$ of the
bound variable all are well-typed. The rule for
meta-variables are built in a similar fashion as the typing rules for
bound variables peeling off type declarations from the meta-context
$\Delta$ until we encounter the meta-variable $X_1$.

\subsection{Definitional Equality}

In this section, we describe a typed $\beta\eta$-equality judgement on
expressions, ordinary substitutions, and meta-substitutions.
We will use the following judgments: 
\[
\begin{array}[h]{l@{~}c@{~}l@{~}c@{~}l@{~}c@{~}lp{8cm}}
\Delta; \Gamma & \vdash & E_1 & \equiv & E_2 & : & F & Expressions $E_1$
and $E_2$ are equal at ``type'' $F$\\
 \Delta; \Gamma & \vdash & \sigma_1 & \equiv & \sigma_2 & : & \Psi &
 Substitutions $\sigma_1$ and $\sigma_2$ are equal at
domain $\Psi$ \\
\Delta  & \vdash & \theta_1 & \equiv & \theta_2 & : & \Delta' & 
  Meta-substitutions $\theta_1$ and $\theta_2$ are equal at domain $\Delta'$ 
\end{array}
\]
The judgement $\Delta; \Gamma \vdash E_1 \equiv E_2 : F$
subsumes the judgements $\Delta; \Gamma \vdash K_1 \equiv K_2 : \lfkind $
  (kinds $K_1$ and $K_2$ are equal), 
$\Delta; \Gamma \vdash A_1 \equiv A_2 : K$ (types $A_1$
and $A_2$ are equal of kind $K$) and 
$\Delta; \Gamma \vdash M_1 \equiv M_2 : A$ (terms $M_1$
and $M_2$ are equal of type $A$).

These judgements are all congruences, \ie, we have equivalence rules
(reflexivity, symmetry, transitivity) and a congruence rule for each
syntactic construction.  For instance, this is one of congruence rule for
substitutions and the type conversion rule:
\[
 \infer{\Delta; \Gamma \vd (\sigma, M) \equiv (\sigma', M')
   \hastype \Psi, A}{\Delta; \Gamma \vd M \equiv M' \hastype [\sigma]A
   & \Delta ; \Gamma \vd A \hastype \lftype 
   & \Delta; \Gamma \vd \sigma \equiv \sigma' \hastype \Psi}
\qquad
\infer{\Delta;\,\Gamma \vd E \equiv E' \hastype F'}{
  \Delta; \Gamma \vd E \equiv E' \hastype F & 
  \Delta; \Gamma \vd F \equiv F' \hastype s}
\]
The remaining rules for definitional equality fall into two classes:
the computational laws for ordinary substitutions (Figure
\ref{fig:cl1}) and the computational laws for meta-substitutions
(Figure \ref{fig:cl2}). Both sets of rules follow the same
principle. They are grouped into identity and composition rules,
propagation and reduction rules. For ordinary substitutions we also
include $\beta$-reduction.  For meta-substitutions, there is no
equivalent $\beta$-reduction rule since we do not support abstraction
over meta-variables. However, we add propagation into ordinary
substitutions.  Note that pushing a meta-substitution inside a
lambda-abstraction or a $\Pi$-type does not require a shift of the
indices, since indices of ordinary bound variables are distinct from
indices of meta-variables and no capture can occur.
There is no reduction for $\esub \sigma \msub \theta M$: an
ordinary substitution cannot in general be pushed past a meta-substitution, it
has to wait for the meta-substitution to be resolved.

\begin{figure}
  \centering
\[
 \begin{array}[h]{c}
\abox{$\beta$-Reduction}
\\[0.75em]
\ru{\Delta; \Gamma, A \der M : B \qquad
    \Delta; \Gamma, A \der B : \lftype \qquad
    \Delta; \Gamma \der N : A
  }{\Delta; \Gamma \der (\lambda M)\,N \equiv \sgsub N M : \sgsub N B} 
\\[1.5em]
\abox{Substitution Propagation: Identity and Composition}
\\[0.75em]
\ru{\Delta; \Gamma \der E : F
  }{\Delta; \Gamma \der \esub {\shiftby 0} E \equiv E : F}
\qquad
\ru{\Delta; \Gamma \der \sigma : \Gamma' \qquad
    \Delta; \Gamma' \der \tau : \Psi \qquad
    \Delta; \Psi \der E : F
  }{\Delta; \Gamma \der \esub \sigma {\esub \tau E} \equiv \esub {\esub
      \sigma \tau} E : \esub {\esub \sigma \tau} F
  } 
\\[1.5em]
\abox{Substitution Propagation: Constants}
\\[0.75em]
\ru{\Delta; \Gamma \der \sigma : \Psi \qquad
  }{\Delta; \Gamma \der \esub \sigma \lftype \equiv \lftype : \lfkind} 
\qquad
\ru{\Delta; \Gamma \der \sigma : \Psi \qquad
    \Delta; \Gamma \der a : K \qquad
  }{\Delta; \Gamma \der \esub \sigma a \equiv a : \esub \sigma K} 
\\[1.5em]
\abox{Substitution Propagation: Variable Lookup}
\\[0.75em]
\ru{\Delta; \Gamma \der \sigma : \Psi \qquad
    \Delta; \Psi \der A : \lftype \qquad
    \Delta; \Gamma \der M : \esub \sigma A 
  }{\Delta; \Gamma \der \esub {\sigma, M} x_1 \equiv M : \esub \sigma A}
\qquad
\ru{\Delta; \Gamma \der x_{n+1} : A
  }{\Delta; \Gamma \der x_{n+1} \equiv \esub{\shiftby 1} x_n : A}
\\[1.5em]
\abox{Substitution Propagation: Pushing into Expression Constructions}
\\[0.75em]
\ru{\Delta; \Gamma \der \sigma : \Psi \qquad
    \Delta; \Psi,A \der F : s
  }{\Delta; \Gamma \der \esubp \sigma {\fun A F} \equiv
      \fun {\esub \sigma A} {\esub {\esub {\shiftby 1} \sigma, x_1} F}
      : s
  } 
\\[1.5em]
\ru{\Delta; \Gamma \der \sigma : \Psi \qquad
    \Delta; \Psi,A \der M : B \qquad
    \Delta; \Psi,A \der B : \lftype 
  }{\Delta; \Gamma \der \esubp {\sigma} {\lambda M} \equiv
      \lambda {\esub {\esub {\shiftby 1} \sigma, x_1} M} :
      \fun {\esub \sigma A} {\esub {\esub {\shiftby 1} \sigma, x_1} B}
  } 
\\[1.5em]
\ru{\Delta; \Gamma \der \sigma : \Psi \qquad
    \Delta; \Psi \der E : \fun A F \qquad
    \Delta; \Psi \der N : A
  }{\Delta; \Gamma \der \esubp \sigma {\app E N} \equiv 
      \app {\esub \sigma E} {\esub \sigma N} : \esub {\sigma, \esub
        \sigma N} F
  }
%
\\[1.5em]
\abox{Substitution Reductions: Pairing and Shifting\hspace{12cm}}
\\[0.75em]
\ru{
    \Delta; \Gamma \der (\sigma,M) : \Psi, \Psi', A \quad
    \length{\Psi'} = n
  }{\Delta; \Gamma \der \esub {\sigma, M} {\shiftby {n+1}} \equiv
    \esub \sigma {\shiftby n} : \Psi } 
\qquad
\ru{\Delta; \Gamma \der \sigma : \Psi' \qquad
    \Delta; \Psi' \der (\tau,M) : \Psi, A
  }{\Delta; \Gamma \der \esub \sigma (\tau, M) 
     \equiv (\esub \sigma \tau, \esub \sigma M) : \Psi, A}
\\[1.5em]
\ru{\Delta \der \Gamma,\Gamma_1,\Gamma_2 \cxt \quad
    \length{\Gamma_1} = m \quad 
    \length{\Gamma_2} = n
  }{\Delta; \Gamma,\Gamma_1,\Gamma_2 \der \esub {\shiftby n}
    {\shiftby m} \equiv \shiftby{n+m} : \Gamma 
  }
\\[1.5em]
\abox{Substitution Reductions: Category Laws\hspace{12cm}}
\\[0.75em]
\ru{\Delta; \Gamma \der \sigma : \Psi
  }{\Delta; \Gamma \der \esub {\shiftby 0} \sigma \equiv \sigma : \Psi}
\qquad
\ru{\Delta; \Gamma \der \sigma : \Psi
  }{\Delta; \Gamma \der \esub \sigma {\shiftby 0}  \equiv \sigma : \Psi}
\\[1.5em]
\ru{\Delta; \Gamma_1 \der \sigma_1 : \Gamma_2 \qquad
    \Delta; \Gamma_2 \der \sigma_2 : \Gamma_3 \qquad
    \Delta; \Gamma_3 \der \sigma_3 : \Gamma_4 
  }{\Delta; \Gamma_1 \der \esub {\sigma_1} \esub {\sigma_2}
    {\sigma_3} \equiv \esub {\esub {\sigma_1}{\sigma_2}} {\sigma_3} : \Gamma_4
   }
\end{array}
\]
  
  \caption{Computational Laws I: $\beta$ and substitutions}
  \label{fig:cl1}
\end{figure}

\begin{figure}
  \centering
\[
 \begin{array}[h]{c}
\abox{Meta-Substitution Propagation: Identity and Composition}
\\[0.75em]
\ru{\Delta; \Gamma \der E : F
  }{\Delta; \Gamma \der \msub {\Shiftby 0} E \equiv E : F}
\qquad
\ru{\Delta \der \theta : \Delta' \qquad
    \Delta' \der {\theta'} : \Delta'' \qquad
    \Delta''; \Gamma \der E : F
  }{\Delta; \msub {\msub \theta {\theta'}} \Gamma \der 
      \msub \theta {\msub {\theta'} E} \equiv \msub {\msub
        \theta {\theta'}} E 
      : \msub {\msub \theta {\theta'}} F
  } 
\\[1.5em]
\abox{Meta-Substitution Propagation: Constants and Ordinary Variables}
\\[0.75em]
\ru{\Delta \der \theta : \Delta' \quad 
    \Delta' \der \Gamma \cxt
  }{\Delta; \msub \theta \Gamma \der 
      \msub \theta \lftype \equiv \lftype : \lfkind} 
\quad
\ru{\Delta \der \theta : \Delta' \quad
    \Delta'; \Gamma \der a : K
  }{\Delta; \msub \theta \Gamma \der \msub \theta a \equiv a : \msub \theta K} 
\quad
\ru{\Delta \der \theta : \Delta' \quad
    \Delta'; \Gamma \der x_n : A 
  }{\Delta; \msub \theta \Gamma \der \msub \theta {x_n} \equiv x_n : \msub \theta A} 
\\[1.5em]
\abox{Meta-Substitution Propagation: Meta-variable Lookup}
\\[0.75em]
\ru{\Delta \der \theta : \Delta' \qquad
    \Delta'; \Gamma \der A : \lftype \qquad
    \Delta; \msub \theta \Gamma \der M : \msub \theta A 
  }{\Delta; \msub{\theta} \Gamma \der \msub {\theta, M} X_1 \equiv M : \msub \theta A}
\qquad
\ru{\Delta; \Gamma \der X_{n+1} : A
  }{\Delta; \Gamma \der X_{n+1} \equiv \msub{\Shiftby 1} X_n : A}
\\[1.5em]
\abox{Meta-Substitution Propagation: Pushing into Expression Constructions}
\\[0.75em]
\ru{\Delta \der \theta : \Delta' \quad
    \Delta'; \Gamma,A \der F : s
  }{\Delta; \msub \theta \Gamma \der \msubp \theta {\fun A F} \equiv
      \fun {\msub \theta A} {\msub \theta F}
      : s
  } 
\quad
\ru{\Delta \der \theta : \Delta' \quad
    \Delta'; \Gamma,A \der M : B \quad
    \Delta'; \Gamma,A \der B : \lftype 
  }{\Delta; \msub \theta \Gamma \der \msubp {\theta} {\lambda M} \equiv
      \lambda {\msub \theta M} :
      \fun {\msub \theta A} {\msub \theta B}
  } 
\\[1.5em]
\ru{\Delta \der \theta : \Delta' \quad
    \Delta'; \Gamma \der E : \fun A F \quad
    \Delta'; \Gamma \der N : A
  }{\Delta; \msub \theta \Gamma \der \msubp \theta {\app E N} \equiv 
      \app {\msub \theta E} {\msub \theta N} : \esub {\shiftby 1, \msub
        \theta N} F
  }
\quad
\ru{\Delta \der \theta : \Delta' \quad
    \Delta'; \Gamma \der \sigma : \Psi \quad
    \Delta'; \Psi \der E : F
  }{\Delta ; \msub \theta \Gamma \der \msub \theta {\esub \sigma E}
      \equiv \esub {\msub \theta \sigma} {\msub \theta E}
      :  \esub {\msub \theta \sigma} {\msub \theta F} 
  }
%
\end{array}
\]
\[
\begin{array}[h]{c}
\abox{Meta-Substitution Propagation: Pushing into Ordinary Substitutions}
\\[0.75em]
\ru{\Delta \der \theta : \Delta' \qquad
    \Delta' \der \Gamma,\Gamma' \cxt \qquad
    \length{\Gamma'} = n
  }{\Delta; \msub \theta \Gamma, \msub \theta {\Gamma'} \der
      \msub \theta {\shiftby n} \equiv \shiftby n : \msub \theta
      \Gamma
  }
\\[1.5em]
\ru{\Delta \der \theta : \Delta' \qquad
    \Delta'; \Gamma \der \sigma : \Psi \qquad
    \Delta'; \Psi \der A : \lftype \qquad
    \Delta'; \Gamma \der M : \esub \sigma A
  }{\Delta; \msub \theta \Gamma \der 
      \msub \theta {(\sigma, M)}
      \equiv (\msub \theta \sigma, \msub \theta M)
      : \msub \theta \Psi, \msub \theta A
  }
\\[1.5em]  
\ru{\Delta \der \theta : \Delta' \quad
    \Delta'; \Gamma \der \tau : \Psi' \quad
    \Delta'; \Psi'  \der \sigma : \Psi
  }{\Delta; \msub \theta \Gamma \der \msub \theta {\esub \tau \sigma}
      \equiv \esub {\msub \theta \tau} {\msub \theta \sigma}
      : \msub \theta \Psi
  }
\quad
\ru{\Delta \der \theta : \Delta' \quad
    \Delta' \der \theta' : \Delta'' \quad
    \Delta''; \Gamma \der \sigma : \Psi
  }{\Delta;\msub{\msub{\theta}{\theta'}} \Gamma \der 
      \msub \theta {\msub{\theta'}\sigma} 
      \equiv \msub{\msub{\theta}{\theta'}} \sigma
      : \msub{\msub{\theta}{\theta'}} \Psi
  }
\end{array}
\]
\[
\begin{array}[h]{c}
\abox{Meta-Substitution Reductions: Pairing and Shifting}
\\[0.75em]
\ru{
    \Delta \der (\theta,M) : \Delta_0, \Delta_0', \cdec \Gamma A \quad
    \length{\Delta_0'} = n
  }{\Delta \der \msub {\theta, M} {\Shiftby {n+1}} \equiv
    \msub \theta {\Shiftby n} : \Delta_0 } 
\qquad
\ru{\Delta \der \theta : \Delta_0' \qquad
    \Delta_0' \der ({\theta'},M) : \Delta_0, \cdec \Gamma A
  }{\Delta \der \msub \theta ({\theta'}, M) 
     \equiv (\msub \theta {\theta'}, \msub \theta M) : \Delta_0, \cdec
     \Gamma A}
\\[1.5em]
\ru{ \der \Delta,\Delta_1,\Delta_2 \mcxt \quad
    \length{\Delta_1} = m \quad 
    \length{\Delta_2} = n
  }{\Delta,\Delta_1,\Delta_2 \der \msub {\Shiftby n}
    {\Shiftby m} \equiv \Shiftby{n+m} : \Delta 
  }
\\[0.75em]
\abox{Meta-Substitution Reductions: Category Laws\hspace{11cm}}
\\[0.75em]
\ru{\Delta \der \theta : \Delta_0
  }{\Delta \der \msub {\Shiftby 0} \theta \equiv \theta : \Delta_0}
\quad
\ru{\Delta \der \theta : \Delta_0
  }{\Delta \der \msub \theta {\Shiftby 0}  \equiv \theta : \Delta_0}
\quad
\ru{\Delta_1 \der \theta_1 : \Delta_2 \quad
     \Delta_2 \der \theta_2 : \Delta_3 \quad
    \Delta_3 \der \theta_3 : \Delta_4 
  }{\Delta_1 \der \msub {\theta_1} \msub {\theta_2}
    {\theta_3} \equiv \msub {\msub {\theta_1}{\theta_2}} {\theta_3} : \Delta_4
   }
\end{array}
\]
  
  \caption{Computational Laws II: Meta-substitution}
  \label{fig:cl2}
\end{figure}

To illustrate the definitional equality rules, we show how to derive
$\Delta ; \Gamma \vdash [\sigma,M]x_{n+1} \equiv [\sigma]x_n :
[\sigma]A$ which also demonstrates that such a rule is
admissible. Transitivity is essential to assemble the following
sub-derivations. We abbreviate the use of congruence by writing ``Cong'', composition by writing ``Comp'', and conversion by writing ``Conv''.

\[
\begin{array}{c}
\multicolumn{1}{l}{\mbox{Step 1:\hspace{14cm}}}\\[-1.25em]
\infer[\text{Comp, Conv}]
 {\Delta ; \Gamma \der  [\sigma,M]x_{n+1} \equiv 
                        [\sigma,M][\shift^1]x_n : [[\sigma,M]\shiftby 1]A}
 {\infer[\text{Cong}]
   {\Delta ; \Gamma \vd [\sigma,M]x_{n+1} \equiv 
                        [\sigma,M][\shift^1]x_n : [\sigma,M][\shiftby 1]A}
   {\infer[\text{Reduction}]
      {\Delta ; \Psi,B \vd x_{n+1} \equiv [\shift^1]x_n : [\shiftby 1]A}
      {\infer[\text{Weakening}]
         {\Delta ; \Psi,B \vd x_{n+1} : \esub{\shiftby 1}A}
            {\Delta ; \Psi \vd x_{n} : A}
      }
   }
 }
\\[0.75em] 
%
\multicolumn{1}{l}{\mbox{Step 2:}}\\[-1.25em]
\infer[\text{Conv}]{\Delta ; \Gamma \vdash  [\sigma, M][\shift^1]x_n \equiv  [[\sigma,M]\shift^1]x_n : [\sigma]A}{
\infer[\text{Comp}]{\Delta ; \Gamma \vdash  [\sigma, M][\shift^1]x_n \equiv  [[\sigma,M]\shift^1]x_n : [[\sigma,M]\shift^1]A}{
\Delta ; \Gamma \vdash \sigma,M : \Psi,B & 
\Delta ; \Psi,B \vdash \shift^1 : \Psi & \Delta ; \Psi \vdash x_n : A
  }
& 
\infer[\text{Cong}]{\Delta ; \Gamma \vdash [[\sigma,M]\shift^1]A \equiv [\sigma]A : \lftype}{
  \begin{array}[h]{c}
    \cal{D} \\ 
    \Delta ; \Gamma \vdash [\sigma,M]\shift^1 \equiv \sigma : \Psi
  \end{array}
}
}
\\[0.85em]
\multicolumn{1}{l}{\mbox{Step 3:}}\\[-1.35em]
\infer[\text{Cong}]{\Delta ; \Gamma \vdash [[\sigma, M]\shift^1]x_n \equiv  [\sigma]x_n : [\sigma]A}{
  \begin{array}[h]{c}
\cal{D} \\
\Delta ; \Gamma \vdash [\sigma, M]\shift^1 \equiv \sigma :   \Psi     
  \end{array}
  } 
\end{array}
\]

where  
\vspace{-0.65cm}
\[
\begin{array}[b]{cc}
  \begin{array}[h]{c}
\\[-2.75em]
\cal{D} =    
  \end{array}
 &
\infer[\text{Transitivity}]{\Delta ; \Gamma \vdash [\sigma,M]\shift^1 \equiv \sigma :\Psi}{
  \infer[\text{Pairing}]{\Delta ; \Gamma \vdash [\sigma,M]\shift^1 = [\sigma]\shift^0 : \Psi}{
          \Delta ; \Gamma \vdash \sigma , M : \Psi,B} 
 & 
  \infer[\text{Category Laws}]{\Delta ; \Gamma \vdash [\sigma]\shift^0 \equiv \sigma : \Psi}{
         \Delta ; \Gamma \vdash \sigma : \Psi}
}  
\end{array}
\]
Similarly, we can show that $\Delta ; \Gamma \vdash \msub{\theta, M} X_{n+1} \equiv \msub{\theta}X_n : \msub{\theta}A$ is admissible.

\subsubsection{Extensionality Laws}
As mentioned earlier, we take into account $\beta$-reductions and
$\eta$-expansions. In particular, we consider $\eta$-rules for
ordinary substitutions as well as meta-substitutions. 

\[
\begin{array}{c}
\ru{\Delta; \Gamma \der M \hastype \fun A B 
  }{\Delta; \Gamma \der M \equiv 
      \lambda (\app {(\esub {\shiftby 1} M)} {x_1}) \hastype \fun A B
  }
\\[1.5em]
\ru{\Delta \der \Gamma,A,\Gamma' \cxt \qquad \length{\Gamma'} = n
  }{\Delta; \Gamma,A,\Gamma' \der \shiftby n \equiv (\shiftby {n+1}, x_{n+1}) : \Gamma,A}
\qquad
\ru{\der \Delta,\cdec \Gamma A,\Delta' \mcxt \qquad \length{\Delta'} = n
  }{\Delta, \cdec \Gamma A,\Delta' \der \mshiftby n \equiv (\mshiftby {n+1},
    X_{n+1}) : \Delta,\cdec \Gamma A}
\end{array}
\]

\LONGVERSION{
\subsubsection{Compatibility and Equivalence Rules}
\[
\begin{array}{c}
\abox{Substitution Congruence\hspace{12cm}}\\[0.75em]
\infer{\Delta; \Gamma,\Gamma' \vd \shiftby n \equiv \shiftby n \hastype \Gamma}
{ \Delta \der \Gamma,\Gamma' \cxt & \length{\Gamma'} = n} 
\hspace{2em}
\infer{\Delta; \Gamma \vd (\sigma, M) \equiv (\sigma', M')
  \hastype \Psi, A}{\Delta; \Gamma \vd M \equiv M' \hastype [\sigma]A
  & \Delta; \Gamma \vd \sigma \equiv \sigma' \hastype \Psi}
\\
\mbox{etc. + equivalences}
\\[0.75em]
\abox{Simultaneous Congruence}\\[0.75em]
\infer{\Delta; \Gamma \vd \lam{}{M} \equiv \lam{}{N} \hastype\Pi A. B}{
 \Delta; \Gamma, A \vd M \equiv N \hastype B}
\\[0.75em]
\infer{\Delta \Gamma \vd M_1\,M_2 \equiv N_1\,N_2 \hastype [\shift^0, M_2]A_1}{
\Delta; \Gamma \vd M_1 \equiv N_1 \hastype \Pi A_2\ldot A_1 & 
\Delta; \Gamma \vd M_2 \equiv N_2 \hastype A_2}
\\[0.75em]
\infer{\Delta; \Gamma \vd  \clo{M_1}{\sigma_1} \equiv
       \clo{M_2}{\sigma_2} \hastype \esub {\sigma_1} A
     }{\Delta; \Gamma \der \sigma_1 \equiv \sigma_2 : \Psi &
       \Delta; \Psi   \der M_1 \equiv M_2 : A
     }
\\[0.75em]
\infer{\Delta; \Gamma \vd  X_n \equiv X_n  \hastype A}{}
\qquad
\infer{\Delta; \Gamma \vd  x_n \equiv x_n  \hastype A}{}
\\[0.75em]
\multicolumn{1}{l}{\mbox{Equivalence}}\\[0.75em]
\infer{\Delta; \,\Gamma \vd M \equiv N \hastype A}{\Delta; \,\Gamma
  \vd N \equiv M \hastype A} \hspace{2em}
\infer{\Delta;\,\Gamma \vd M \equiv M' \hastype A}{\Delta;\,\Gamma \vd
  M \equiv N \hastype A & \Delta;\,\Gamma \vd N \equiv M' \hastype
  A}
\\[0.75em]
\multicolumn{1}{l}{\mbox{Type conversion}}\\[0.75em]
\infer{\Delta;\,\Gamma \vd M \equiv N \hastype B}{\Delta;\,\Gamma \vd
  M \equiv N \hastype A & \Delta; \Gamma \vd A \equiv B \hastype
  \lftype}
\\[0.75em]
\abox{Family congruence}\\[0.75em]
\infer{\Delta; \Gamma \vd a \equiv a \hastype E}{a \oftp E \mbox{ in } \Sigma}
\hspace{2em}
\infer{\Delta; \Gamma \vd A\,M \equiv B\,N \hastype [\shift^0, M]K}{
       \Delta; \Gamma \vd A \equiv B \hastype \Pi C. K & 
       \Delta; \Gamma \vd M \equiv N \hastype C} \\[1.5em]

\infer{\Delta; \Gamma \vd \Pi A_1. A_2 \equiv \Pi B_1. B_2 \hastype \lftype}{
       \Delta; \Gamma \vd A_1 \equiv B_1 \hastype \lftype & 
       \{\Delta; \Gamma \vd A_1 \hastype \lftype\} & 
       \Delta; \Gamma, A_1 \vd A_2 \equiv B_2 \hastype \lftype}
\\[0.75em]
\abox{Family equivalence}\\[0.75em]
\infer{\Delta; \Gamma \vd B \equiv A \hastype K}{
  \Delta; \Gamma \vd A \equiv B \hastype K} 
\hspace{2em}
\infer{\Delta; \Gamma \vd A \equiv C \hastype K}{
  \Delta; \Gamma \vd A \equiv B \hastype K & \Delta; \Gamma \vd B \equiv C \hastype K}
\\[0.75em]
\abox{Kind conversion}\\[0.75em]
\infer{\Delta; \Gamma \vd A \equiv B \hastype L}{
  \Delta; \Gamma \vd A \equiv B \hastype K & 
  \Delta; \Gamma \vd K \equiv L \hastype \lfkind}
\\[0.75em]
\abox{Kind congruence\hspace{12cm}}\\[0.75em]
\infer{\Delta; \Gamma \vd \lftype \equiv \lftype \hastype \lfkind}{}
\\[0.75em]
\infer{\Delta; \Gamma \vd \Pi A. K \equiv \Pi B. L \hastype \lfkind}{
       \Delta; \Gamma \vd A \equiv B \hastype \lftype & 
       \{\Delta; \Gamma \vd A \hastype \lftype\} & 
       \Delta; \Gamma, A \vd K \equiv L \hastype \lfkind}
\\[0.75em]
\abox{Kind equivalence}\\[0.75em]
\infer{\Delta; \Gamma \vd L \equiv K \hastype \lfkind}{
  \Delta; \Gamma \vd K \equiv L \hastype \lfkind} 
\hspace{2em}
\infer{\Delta; \Gamma \vd K \equiv L \hastype \lfkind}{
  \Delta; \Gamma \vd K \equiv L' \hastype \lfkind & 
  \Delta; \Gamma \vd L' \equiv L \hastype \lfkind} 
\end{array}
\] 
} 

\subsection{Properties}
Next, we prove some standard properties about the presented type
assignment system. First, we show that contexts are indeed
well-formed. 
%
\begin{lemma}[Context well-formedness] \label{lem:cwf} \bla
  \begin{enumerate}
  \item If $\Delta,\Delta' \der J$ or $\Delta,\Delta'; \Gamma \der J$ 
        then $\der \Delta \mcxt$.
  \item If $\Delta \der \theta : \Delta'$ or $\Delta \der \theta
    \equiv \theta' : \Delta'$ then $\der \Delta' \mcxt$.
  \item If $\Delta; \Gamma,\Gamma' \der J$ then $\Delta \der \Gamma \cxt$.
  \item If $\Delta; \Gamma \der \sigma : \Psi$ or $\Delta ; \Gamma
    \der \sigma \equiv \sigma' : \Psi$ then $ \Delta \der \Psi \cxt$.
  \end{enumerate} 
The height of the output derivation is bounded by the height of the
input derivation, in all cases.
\end{lemma}
\begin{proof}
  By simultaneous induction over all judgments.
\end{proof}

The following inversion theorem for typing is standard for PTSs and
are necessary due to the type conversion rule which makes inversion a
non-obvious property.  It allows us to classify expressions into terms,
types, kinds, and the sort $\lfkind$.  
We write $\Delta;\Gamma \der E \equiv E'$ if there exists a sort $s$
such that $\Delta;\Gamma \der E \equiv E' : s$.
\begin{theorem}[Inversion of typing] \label{thm:inv} \bla
\begin{enumerate}
\item There is no derivation of $\Delta;\Gamma \der \lfkind : E$.
\item If $\Delta;\Gamma \der \lftype : E$ then $E = \lfkind$.      
\item If $\Delta;\Gamma \der a : E$ then $\Delta;\Gamma \der E \equiv
  \Sigma(a)$.
\item If $\Delta;\Gamma \der \fun A E : F$ then 
   $\Delta;\Gamma \der A : \lftype$ and 
   $\Delta; \Gamma,A \der E : F$ and 
   either $F = \lfkind$ or $\Delta; \Gamma \der F \equiv \lftype$. 
\item If $\Delta;\Gamma \der x_{n+1} : A$ then $\Gamma =
  \Gamma_1,A',\Gamma_2$ with $\length{\Gamma_2} = n$ and
  $\Delta;\Gamma \der A \equiv \esub{\shiftby {n+1}} {A'}$. 
\item If $\Delta; \Gamma \der X_{n+1} : A$ then $\Delta =
  \Delta_1,\cdec{\Gamma'}{A'},\Delta_2$ with $\length{\Delta_2} = n$ and
  $\Gamma = \msub{\Shiftby{n+1}}{\Gamma'}$ and
  $\Delta;\Gamma \der A \equiv \msub{\Shiftby {n+1}} {A'}$.
\item If $\Delta; \Gamma \der \lambda M : C$ then there are $A,B$ such
  that $\Delta; \Gamma \der C \equiv \fun A B$ 
  and
  $\Delta; \Gamma \der A : \lftype$ and \\
  $\Delta; \Gamma, A \der B : \lftype$ 
  and
  $\Delta; \Gamma, A \der M : B$.
\item If $\Delta; \Gamma \der \app E N : C$ then there are $A,F$ such
  that $\Delta; \Gamma \der E : \fun A F$ and 
  $\Delta; \Gamma \der N : A$ \\and
  $\Delta; \Gamma \der C \equiv \sgsub N F$. 
\item If $\Delta; \Gamma \der \esub \sigma E : F$ then there are $\Psi,{F'}$
  such that $\Delta; \Gamma \der \sigma : \Psi$ and $\Delta; \Psi \der
  E : {F'}$ \\ and $\Delta; \Gamma \der E \equiv \esub \sigma {F'}$.
\item If $\Delta; \Gamma \der \msub \theta E : F$ then there are
  $\Delta',\Gamma', {F'}$
  such that $\Delta \der \theta : \Delta'$ and 
  $\Delta; \Gamma' \der E : {F'}$ and 
  $\Gamma = \msub \theta {\Gamma'}$ \\and 
  $\Delta; \Gamma \der F \equiv \msub \theta {F'}$. 
\end{enumerate}
\end{theorem}
\begin{proof}
  By induction on the typing derivation, peeling off the type
  conversion steps and combining them with transitivity.
\end{proof}
Expression $E$ is a \emph{kind} if $\Delta;\Gamma \der E : \lfkind$ for some
$\Delta,\Gamma$, it is a \emph{type family} if $\Delta;\Gamma \der E : K$ for
some kind $K$ and some $\Delta,\Gamma$, and it is a \emph{term} if
$\Delta;\Gamma \der E : A$ for some $A,\Delta,\Gamma$ with
$\Delta;\Gamma \der A : \lftype$.

The following inversion statement for meta-variables under a
substitution is crucial for the correctness of algorithmic equality
(Sec.~\ref{sec:aleq}) and bidirectional type checking (Sec.~\ref{sec:bidir}).
\begin{corollary} \label{cor:invmeta}
  If $\Delta; \Gamma \der \esub \sigma {X_{m}} : A$ then 
$\Delta = \Delta_1, \cdec \Psi {A'}, \Delta_2$ with 
$\length{\Delta_2} = m-1$ and 
$\Delta; \Gamma \der \sigma : \msub {\Shiftby {m}} \Psi$ and 
$\Delta; \Gamma \der A \equiv \dsub \sigma {\Shiftby {m}} {A'}$.  
\end{corollary}

\begin{theorem}[Syntactic Validity]
  \label{thm:synval} \bla
\begin{enumerate}
\item If $\Delta; \Gamma \der E : F$ or $\Delta ; \Gamma \der E_1
  \equiv E_2 : F$ then $\Delta; \Gamma \der F : s$
  for some sort.
\item If $\Delta; \Gamma \der E \equiv E' : F$ then $\Delta; \Gamma \der E : F$
  and $\Delta; \Gamma \der E' : F$. 
\item If $\Delta \der \theta \equiv \theta' : \Delta'$ then $\Delta \der \theta : \Delta'$
  and $\Delta \der \theta' : \Delta'$. 
\item If $\Delta; \Gamma \der \sigma \equiv \sigma' : \Psi$ then $\Delta; \Gamma \der \sigma : \Psi$
  and $\Delta; \Gamma \der \sigma' : \Psi$. 
\end{enumerate}
\end{theorem}
\begin{proof}
  By simultaneous induction over all judgments.
\end{proof}

\section{Evaluation and Algorithmic Equality}
\label{sec:algo}

In this section, we define a weak head normalization strategy together
with algorithmic equality. The goal is to treat ordinary substitutions
and meta-substitutions lazily; in particular, we aim to postpone
shifting of substitutions until necessary.  
For the treatment of LF, an untyped algorithmic equality is
sufficient.  The design of the algorithm follows Coquand 
\cite{coquand:conversion} with refinements from joint work with the
first author \cite{abelCoquand:lfsigma}.  In this article, we only
show soundness of the algorithm; completeness can be proven using
techniques of the cited works.  However, an adaptation to de Bruijn
style and explicit substitutions is necessary; 
we leave the details to future work%
\LONGVERSION{, a sketch can be found in the appendix}.

We first characterize
our normal forms by defining normal and neutral expressions where
expressions include terms, types, and kinds. 
Normal forms are exactly the expressions we can type-check with a
bidirectional algorithm (see Section~\ref{sec:bidir}).  Note that type
checking normal forms is sufficient in practice, since the input to
the type checker, written by a user, is almost always in
$\beta$-normal form (or can be turned into normal form by introducing
typed let-definitions).

Normal substitutions are
built out of normal expressions. However, it is worth keeping in mind
that our typing rules will ensure that they only contain terms and
not types, since we do not support type-level variables. Our normal
forms are only $\beta$-normal, not necessarily $\eta$-long. Only
meta-variables are associated with an ordinary normal substitution,
all other closures have been eliminated.   

\[
\begin{array}{llrl}
\mbox{Normal expressions} &
  V & ::= & s \mid \fun V V' \mid \lambda V \mid U 
\\
\mbox{Neutral expressions} &
  U & ::= & a \mid x_n \mid \esub \nu {X_n} \mid \app U V 
\\[0.75em]
\mbox{Normal substitutions} &
  \nu & ::= & \shiftby n \mid (\nu, V)
\end{array}
\]

Next, we define weak head normal forms (whnf). 
Since we want to treat
ordinary substitutions and meta-substitutions lazily and in particular
want to postpone the complete computation of their compositions, 
we cannot require
that substitutions and meta-substitutions are already in normal
form. Hence, we introduce environments $\rho$ for ordinary
substitutions and similarly meta-substitutions $\eta$ for describing
substitutions and meta-substitutions that are in weak head
normal form. 
Closures are expressions $E$ in an environment $\rho$ and
a meta-environment $\eta$.  It is convenient to also treat variables
$x_n$ as closures. These arise when stepping under a binder in type
and equality checking and are the synonym of Coquand's generic values
\cite{coquand:type}.

\[
\begin{array}{llrl}
\mbox{Weak head normal forms} &
  W & ::= & \lftype
    \mid \dsub \rho \eta {\fun A B} 
    \mid \dsub \rho \eta {\lambda M}
    \mid H \\
\mbox{Neutral weak head normal forms} &
  H & ::= & a \mid x_n \mid \esub \rho {X_n} \mid \app H L 
\\
\mbox{Closures} &
  L & ::= & x_n \mid \dsub \rho \eta E 
\\[0.75em] 
\mbox{Environments} &
  \rho & ::= & \shiftby n \mid (\rho, L) \mid \esub {\shiftby n} \rho \\
\mbox{Meta-environments} &
  \eta & ::= & \Shiftby n \mid (\eta, M) \\
\end{array}
\]

Our weak head normal forms and closures combine substitutions and
meta-substitutions
and our whnf-reduction
strategy simultaneously treats substitutions and meta-substitutions. 
Instead of coupling expressions with two suspended substitutions, we
could have introduced a joint simultaneous substitutions and closures
built with them. The path taken in this paper builds on the individual
substitution operations instead of defining a new joint substitution
operation. To clarify the nature and the interplay of ordinary
substitutions and meta-substitutions it is helpful to
consider the typing rule of closures $\dsub \rho \eta E$:

\[
\begin{array}{c}
\infer{\Delta; \Psi \vdash \dsub \rho  \eta E :  \dsub \rho \eta F}
{  \Delta; \Psi \vdash \rho : \msub{\eta}\Psi' & 
 \Delta \vdash \eta : \Delta' & 
 \Delta' ; \Psi' \vdash E : F 
}
%
%
%
\end{array}
\]

Intuitively, this means to obtain an expression $E'$ which makes actually
sense in $\Delta$ and $\Psi$, we first compute $\msub{\eta}E$ and
subsequently apply the ordinary substitution $\rho$ to
arrive at $E' \equiv ([\rho]\msub{\eta}E)$.

\paradot{Shift propagation}\label{abbrev-sclo}
While we treat shifts in the environment as an explicit
operation---to avoid a traversal when lifting an environment under a
binder---, shifting a closure or a neutral weak head normal form can
be implemented inexpensively.
Let shifting $\shiftClos n L$ of a closure $L$ be defined by
$\shiftClos n {x_m} = x_{n+m}$ and $\shiftClosp n {\dsub \rho \eta E} =
\dsub {\shiftEnv n \rho} \eta E$.  
It is extended to shifting 
of neutral weak head normal forms $H$ by 
$\shiftNep n {\app H L} = (\shiftNe n H) ~ (\shiftClos n L)$ and
$\shiftNep n {\esub \rho {X_m}} = \esub {\esub {\shiftby n} \rho}
{X_m}$ and $\shiftNe n a = a$.

\subsection{Weak head evaluation}

Our weak head evaluation strategy will postpone
propagation of substitutions into an expression until necessary. 
Treating substitutions lazily seems to be beneficial as also
supported by the experimental analysis on lazy vs eager reduction
strategies for substitutions by Nadathur and his collaborators
\cite{liangNadathurQi:jar05}. We present the algorithm for weak head
normalization in Figure \ref{fig:sub2env}.
We define a function $\twhnf~L$ where $L$ is either a variable $x_n$ or
a proper closure $\dsub \rho \eta E$. The function $\twhnf\/$ is then
defined recursively on $E$.  

\begin{figure}
  \centering
\[
\begin{array}{lll}
\aboxiii{Meta-substitution evaluation $\wmsub \eta \theta$ 
  computes the meta-environment form of $\msub \eta \theta$.} 
\\[0.5em]
  \wmsub {\Shiftby m}{\Shiftby n}    & = & \Shiftby{m+n} \\
  \wmsub {(\eta, M)}{\Shiftby {n+1}} & = & \wmsub \eta {\Shiftby n} \\
  \wmsub {\eta}{(\theta, M)}         & = & (\wmsub \eta \theta, \msub
    \eta M) \\
  \wmsub {\eta}{\msub \theta {\theta'}} & = & \wmsub {(\wmsub \eta
    \theta)} {\theta'}  
\\[0.5em]
\aboxiii{Substitution evaluation $\wsub \rho \eta \sigma$ 
computes the environment form of $\dsub \rho \eta \sigma$.} 
\\[0.5em]
  \wsub {(\shiftEnv k \rho)} \eta \sigma & = & 
    \shiftEnvp k {\wsub \rho \eta \sigma} \\
  \wsub \rho \eta {\shiftby 0} & = & \rho \\  
  \wsub {\shiftby k} \eta {\shiftby n}   & = & \shiftby {k + n}  \\
  \wsub {(\rho, L)} \eta {\shiftby{n+1}} & = & \wsub \rho \eta {\shiftby n} \\
  \wsubp \rho \eta {\sigma, M} & = & (\wsub \rho \eta \sigma,
    \dsub \rho \eta M) \\   
  \wsubp \rho \eta {\esub \sigma \tau} & = & \wsub {(\wsub \rho
    \eta \sigma)} \eta \tau \\  
  \wsubp \rho \eta {\msub \theta \sigma} & = & \wsub \rho {(\wmsub
    \eta \theta)} \sigma 
\\[0.5em]
\aboxiii{Meta-variable lookup $\wmlookup \eta m$ 
  retrieves the binding of $X_m$ in meta-environment $\eta$.}
\\[0.5em]
  \wmlookup  {\Shiftby n} m & = & X_{n+m} \\
  \wmlookupp {\eta, E}    1 & = & E \\
  \wmlookupp {\eta, E}{m+1} & = & \wmlookup \eta m
\\[0.5em]
\aboxiii{Variable lookup $\wlookup \rho m$ 
computes the closure form of $\esub \rho {x_m}$.} 
\\[0.5em]
  \wlookup  {\shiftby n} m & = & x_{n+m} \\
  \wlookupp {\rho, L}  1   & = & L \\
  \wlookupp {\rho, L}             {m+1} & = & \wlookup \rho m \\
  \wlookupp {\esub {\shiftby n} \rho} m & = & \shiftClosp n {\wlookup \rho m} 
\\[0.5em]
\aboxiii{Weak head evaluation $\twhnf~L$ computes the weak head normal
form of closure $L$.}
\\[0.5em]
  \twhnf~x_m         & = & x_m \\
  \whnf  \rho \eta s & = & s \\
  \whnf  \rho \eta a & = & a\\
  \whnf  \rho \eta {x_m} & = & \twhnfp {\wlookup \rho m} \\  
  \whnf  \rho \mId {X_m} & = & \esub \rho {X_m} \\
  \whnf  \rho \eta {X_m} & = & \whnf \rho \mId (\wmlookup \eta m) \\
  \whnfp \rho \eta {\fun A E} & = & \dsubp \rho \eta {\fun A E} \\
  \whnfp \rho \eta {\lambda M} & = & \dsubp \rho \eta {\lambda M} \\
  \whnfp \rho \eta {\app M N}  & = & 
    \wapp{(\whnf \rho \eta M)}{\dsub \rho \eta N}\\
  \whnf \rho \eta {\esub \sigma M} & = & 
    \whnf {\wsub \rho \eta \sigma} \eta M\\
  \whnf \rho \eta {\msub \theta M} & = & \whnf \rho {\wmsub \eta \theta} M
\\[0.5em]
\aboxiii{Evaluating application $\wapp W L$ computes the weak head
  normalform of $\app W L$.}
\\[0.5em]
  \wapp {\dsubp \rho \eta {\lambda M}} L & = & \whnf {\rho, L} \eta
  M \\
  \wapp H L & = & \app H L 
\end{array}
\]

  \caption{Weak head evaluation}
  \label{fig:sub2env}
\end{figure}

To support the lazy evaluation 
of substitutions, our weak head
normalization algorithm relies on the definition of two functions,
namely $\wmsub \eta \theta$ and $\wsub \rho \eta \sigma$. Both
functions are defined recursively over the last argument, i.e., $\twmsub$
is inductively defined over $\theta$ and $\twsub$ 
{is} inductively defined over $\sigma$. When we encounter a closure of $[\sigma]\tau$
(or $\msub{\theta}\theta'$ resp.), we compute first the environment
corresponding to $\sigma$ and subsequently we compute the environment
for $\tau$. This strategy allows us to avoid unnecessary shifting of
de Bruijn indices. 

In addition, $\twhnf$ relies on a $\tlookup$ function to retrieve the $i$-th
element of a substitution which corresponds to the index $i$. Such
lookup functions are defined for both, ordinary variables and
meta-variables. 

Next, we prove that types are preserved when computing weak head
normal forms and that the computation is sound with regard to the
specification of definitional equality.  Note that at this point
termination is only clear for the lookup and substitution evaluation
functions.  For $\twhnf$ and evaluating application $@$, soundness is
conditional on termination. 

\begin{theorem}[Subject reduction] \label{thm:sr}
Let $\Delta \der \eta : \Delta'$.
\begin{enumerate}

\item If $\Delta' \der \theta : \Delta''$ then $\Delta \der 
  \wmsub \eta \theta 
  \equiv 
  \msub \eta \theta 
  : \Delta''$.

\item If $\Delta'; \Psi \der \sigma : \Psi'$ and $\Delta; \Gamma
  \der \rho : \msub \eta \Psi$ 
  then 
  $\Delta; \Gamma \der 
   \wsub \rho \eta \sigma 
   \equiv 
   \dsub \rho \eta \sigma
   : \msub \eta {\Psi'}$. 

\item If  $\Delta'; \Psi \der X_m : A$ 
  then
  $\Delta; \msub \eta \Psi \der \wmlookup \eta m \equiv \msub \eta {X_m}
   : \msub  \eta A$.

\item If $\Delta;\Psi \der x_m : A$ and 
  $\Delta; \Gamma \der \rho : \Psi$
  then $\Delta; \Gamma \der \wlookup \rho m \equiv 
     \esub \rho {x_m} :  \esub \rho A$.

\item Let $\Delta'; \Psi \der E : F$ and 
  $\Delta; \Gamma \der \rho : \msub \eta \Psi$.\\
  If $\whnf \rho \eta E$ is defined then $\Delta; \Gamma \der \whnf
  \rho \eta E \equiv \dsub \rho \eta E : \dsub \rho \eta F$.

\item Let $\Delta; \Gamma \der W : \fun A F$ and $\Delta; \Gamma \der
  L : A$.  If\/ $\wapp W L$ is defined then $\Delta; \Gamma \der \wapp W
  L \equiv \app W L : \sgsub L F$.

\end{enumerate}  
\end{theorem}
\begin{proof}
  Each by induction on the trace of the function and inversion on the
  typing derivations, the first four statements in isolation and the
  remaining two simultaneously.
\end{proof}

\subsection{Algorithmic equality} 
\label{sec:aleq}

Building on the weak head normalization algorithm introduced in the
previous section, we now give an algorithm for deciding equality of
expressions. This is a key piece in the bi-directional type checking
algorithm which we present in Section \ref{sec:bidir}.  Two
closures, where 
$L = \dsub \rho \eta E$ and $L' = \dsub {\rho'} {\eta'} {E'}$, are
algorithmically equal 
if their weak head normal forms are related, i.e., 
$\whnf \rho \eta E \alw \whnf {\rho'} {\eta'} {E'}$.

As we check
that two expressions are equal, we lazily normalize them using our
weak head normalization algorithm from the previous section and our
algorithmic equality algorithm alternates between applying a whnf step
and actually comparing two expressions or substitutions. 

%


The actual equality algorithm is defined using three mutual recursive
judgments. 1) checking that two expressions in whnf are equal 2)
checking that two neutral weak head normal forms are equal and 3) checking
that two environments, i.e., ordinary substitutions in whnf, are equal.
\[
\begin{array}{ccl}
  W \alw W'  & & \mbox{weak head normal forms $W,W'$ are
    algorithmically equal} \\
  H \aln H'  & & \mbox{neutral weak head normal forms $H,H'$ are algorithmically equal} \\
  \alreq k {k'} {\rho}{\rho'} & & \mbox{environments $\rho,\rho'$ are
    algorithmically equal under shifts by $k,k'$ resp.} \\
\end{array}
\]

Many of the algorithmic equality rules are straightforward and
intuitive, although a bit veiled by the abundance of explicit shifting
that comes with de Bruijn style. 
When checking whether two meta-variables are equal,
we need to make sure that respective environments are equal. 
When we check whether two
lambda-abstractions are equal, we must 
lift their environments under
the lambda-binding.  This amounts to shifting them by one and extending
them with a binding for the first variable.
To handle eta-equality, we eta-expand the neutral weak head normal form
$H$ on the fly when comparing it to a lambda-closure. 

Comparing two environments for equality simply recursively analyzes the
substitutions. In addition, we handle just-in-time
eta-expansion on the level of
substitutions (see the last two rules).

\begin{gather*}
\mbox{Algorithmic equality of neutral weak head normal forms.\hspace{6.5cm}}\\
  \ru{}{a \aln a}
\quad
  \ru{}{x_m \aln x_m}
\quad
  \ru{\alreq 0 0 \rho {\rho'}
    }{\esub \rho {X_m} \aln \esub {\rho'} {X_m}}
\quad
  \ru{H \aln H' \quad 
      \twhnf~ L \alw \twhnf~ L' 
    }{\app H L \aln \app {H'} {L'}}
\\[0.75em]
\mbox{Algorithmic equality of weak head normal forms.\hspace{7.5cm}}\\
  \ru{}{s \alw s}
\qquad
  \ru{\whnf \rho \eta A \alw \whnf {\rho'}{\eta'}{A'}
      \qquad
      \whnf {\lift\rho} \eta B \alw \whnf {\lift{\rho'}} {\eta'}{B'}
    }{\dsubp \rho \eta {\fun A B} \alw 
      \dsubp {\rho'}{\eta'}{\fun {A'}{B'}}}
\\[0.75em]
  \ru{H \aln H'  
    }{H \alw H}
\qquad
  \ru{\whnf {\lift\rho} \eta M \alw \whnf {\lift{\rho'}} {\eta'}{M'}
    }{\dsubp \rho \eta {\lambda M} \alw 
      \dsubp {\rho'}{\eta'}{\lambda M'}}
\end{gather*}
\begin{gather*}
  \ru{\whnf {\lift\rho} \eta M \alw 
      \app{(\shiftNe 1 H)}{x_1}  
    }{\dsubp \rho \eta {\lambda M} \alw H}
\qquad
  \ru{\app{(\shiftNe 1 H)}{x_1}
      \alw \whnf {\lift\rho} \eta M 
    }{H \alw \dsubp \rho \eta {\lambda M}}
\end{gather*}
\begin{gather*}
\mbox{Algorithmic equality of environments.\hspace{10cm}}\\
  \ru{k + n = k' + n'
    }{\alreq k {k'}{\shiftby n}{\shiftby {n'}}}
\qquad
  \ru{\alreq{k+n}{k'}\rho{\rho'}
    }{\alreq k{k'}{\shiftEnv n \rho}{\rho'}}
\qquad
  \ru{\alreq{k}{k'+n'}\rho{\rho'}
    }{\alreq k{k'}{\rho}{\shiftEnv {n'}{\rho'}}}
\\[0.75em]
  \ru{\alreq k{k'}\rho{\rho'} \qquad
      \twhnfp{\shiftClos k L} \alw \twhnfp{\shiftClos {k'}{L'}}
    }{\alreq{k}{k'}{(\rho, L)}{(\rho', L')}}
\\[0.75em]
  \ru{\alreq k {k'} \rho {\shiftby{n'+1}} \qquad
      \twhnfp{\shiftClos k L} \alw x_{k'+n'+1}
    }{\alreq k {k'} {(\rho, L)}{\shiftby {n'}}}
\qquad
  \ru{\alreq k {k'} {\shiftby{n+1}} {\rho'} \qquad
      x_{k+n+1} \alw \twhnfp{\shiftClos {k'} {L'}}
    }{\alreq k {k'} {\shiftby n}{(\rho', L')}}
\end{gather*}

\begin{theorem}[Soundness of algorithmic equality]\label{thm:soundeq} \bla
\begin{enumerate}
\item If $H \aln H'$ and $\Delta; \Gamma \der H : F$ and $\Delta;
  \Gamma \der H' : F'$ then $\Delta; \Gamma \der F \equiv F'$ and
  $\Delta; \Gamma \der H \equiv H' : F$. 
\item If $W \alw W'$ 
  and $\Delta; \Gamma \der W : F$ 
  and $\Delta; \Gamma \der W' : F$ 
  then $\Delta; \Gamma \der W \equiv W' : F$.
\item If $\alreq k {k'} \rho {\rho'}$ 
  and $\Delta;\Gamma \der \shiftEnv k \rho : \Psi$
  and $\Delta;\Gamma \der \shiftEnv {k'}{\rho'} : \Psi$
  then
  $\Delta;\Gamma \der \shiftEnv k \rho \equiv \shiftEnv {k'}{\rho'} : \Psi$.
\end{enumerate}
\end{theorem}
\begin{proof}
  Simultaneously by induction on the derivation of algorithmic
  equality and inversion on the typing.
\end{proof}

\section{Bidirectional Type Checking} 
\label{sec:bidir}
In this section, we show how to use our explicit substitution calculus
to type-check expressions. As mentioned in the introduction,
accumulating substitution walks in type-checking is one of the key
applications of this work. We only describe the algorithm and leave
its theoretical properties for future work.

We design the algorithm in a bidirectional way
\cite{coquand:type,abelCoquand:lfsigma} which allows us to omit
type annotations at lambda-abstractions. We use the following three judgments:

\[
\begin{array}{ll}
  \Delta; \Gamma \der V \jchk s & \mbox{Type normal form $V$ checks against
    sort $s$} \\
  \Delta; \Gamma \der V \jchk L & \mbox{Normal form $V$ checks against
    ``type'' closure $L$} \\
  \Delta; \Gamma \der U \jinf L & \mbox{The type of neutral normal
    form $U$ is inferred as closure $L$} 
\\[0.75em]
  \Delta; \Gamma \der \nu \jchk \Psi & \mbox{Normal substitution
    $\nu$ checks against domain $\Psi$} 
\end{array}
\]
In these judgements, $\Gamma$ is a list of type closures $L$.  On
$\Delta$ we pose no restrictions; an entry $\cdec \Psi A$ of $\Delta$
is as before a list of type expressions $\Psi$ and a type expression $A$.
 

\noindent
Inferring the type of neutral normal forms $U$.
\begin{gather*}
  \ru{}{\Delta; \Gamma \der a \jinf \dsub \sid \mId \Sigma(a)}
\qquad 
  \ru{\Delta; \Gamma \der U \jinf L \qquad
      \twhnf~L = \dsubp \rho \eta {\fun A B} \qquad
      \Delta; \Gamma \der V \jchk \dsub \rho \eta A
    }{\Delta; \Gamma \der \app U V \jinf \dsub {\rho, V} \eta B}
\\[0.75em]
  \ru{\length{\Gamma'} = n
    }{\Delta; \Gamma, L, \Gamma' \der x_{n+1} \jinf 
      \shiftClos {n+1} {L} 
      }
\qquad 
  \ru{\Delta = \Delta_1,\cdec \Psi A,\Delta_2 \qquad
      \length{\Delta_2} = n \qquad
      \Delta; \Gamma \der \nu \jchk \dsub{\sid}{\Shiftby{n+1}}\Psi
    }{\Delta; \Gamma \der \esub \nu {X_{n+1}} \jinf 
        \dsub \nu {\Shiftby {n+1}} A} 
\end{gather*}
Checking the type of normal forms $V$.
\begin{gather*}
  \ru{
      \twhnf~L = \dsubp \rho \eta {\fun A B} \quad
      \Delta; \Gamma, \dsub \rho \eta A  \der V \jchk
         \dsub {\shiftby 1 \rho, x_1} \eta B
    }{\Delta; \Gamma \der \lambda V \jchk L}
\quad
  \ru{\Delta; \Gamma \der U \jinf L \quad
      \twhnf~L \alw \twhnf~L'
    }{\Delta; \Gamma \der U \jchk L'}
\end{gather*}
Checking well-formedness of types and kinds $V$.
\begin{gather*}
   \ru{
     }{\Delta; \Gamma \der \lftype \jchk \lfkind}
\quad
  \ru{\Delta; \Gamma \der V \jchk \lftype \quad
      \Delta; \Gamma,  \dsub \sid \mId V \der V' \jchk s
    }{\Delta; \Gamma \der \fun V V' \jchk s} 
\quad
  \ru{\Delta; \Gamma \der U \jinf L \quad
      \twhnf~L = \lftype
    }{\Delta; \Gamma \der U \jchk \lftype}
\end{gather*}
Checking normal substitutions $\nu$.  In this judgement $\Delta;
\Gamma \der \nu \jchk \Psi$, the context $\Psi$ is also in closure form. 
\begin{gather*}
  \ru{\length\Gamma = n
    }{\Delta; \Gamma \der \shiftby n \jchk \cempty}
\qquad
  \ru{\Delta; \Gamma \der \nu \jchk \Psi \qquad
      \Delta; \Gamma \der V \jchk L
    }{\Delta; \Gamma \der (\nu, V) \jchk \Psi,L}
\end{gather*}

\section{Conclusion}
We have presented an explicit substitution calculus together with
algorithms for weak head normalization, definitional equality, and
bi-directional type checking where both ordinary variables and
meta-variables are modelled using de Bruijn indices and both kinds of
substitutions are handled lazily and simultaneously. 

We also have proven subject reduction and soundness of the
definitional equality algorithm.\LONGVERSION{
A sketch of the normalization proof,
which guarantees that the described algorithm is complete, can be
found in the appendix.}
Finally, we describe a bi-directional
type-checking algorithm which treats ordinary substitutions and
meta-substitutions at the same time. In the future, we plan to 
\SHORTVERSION{prove completeness of algorithmic equality and type
  checking and to }adapt
the presented explicit substitutions in the implementation of the
programming and reasoning environment Beluga.

\bibliographystyle{eptcs}

\LONGVERSION{

\appendix
 
\input{normalization}




}

\end{document}